\documentclass[a4paper,11pt,reqno]{amsart}
\usepackage[utf8]{inputenc}
\usepackage{amsfonts}
\usepackage{amsmath}
\usepackage{amssymb}
\usepackage{amsthm}
\usepackage{verbatim}
\usepackage{color}
\usepackage{hyperref}
\usepackage{epstopdf}

\newtheorem{theorem}{Theorem}[section]

\newtheorem{lemma}[theorem]{Lemma}
\newtheorem{conjecture}[theorem]{Conjecture}
\newtheorem{proposition}[theorem]{Proposition}

\theoremstyle{remark}

\numberwithin{equation}{section}

\title[Joint moments of random symplectic and orthogonal matrices]{Joint moments of derivatives of characteristic polynomials of random symplectic and orthogonal matrices}

\author{Julio C. Andrade}
\address{Department of Mathematics, University of Exeter, Exeter, EX4 4QF, United Kingdom}
\email{j.c.andrade@exeter.ac.uk}

\author{Christopher G. Best}
\address{Department of Mathematics, University of Exeter, Exeter, EX4 4QF, United Kingdom}
\email{cgb212@exeter.ac.uk}

\date{\today}


\begin{document}

\begin{abstract}
We investigate the joint moments of derivatives of characteristic polynomials over the unitary symplectic group $Sp(2N)$ and the orthogonal ensembles $SO(2N)$ and $O^-(2N)$. We prove asymptotic formulae for the joint moments of the $n_1$-th and $n_2$-th derivatives of the characteristic polynomials for all three matrix ensembles. Our results give two explicit formulae for each of the leading order coefficients, one in terms of determinants of hypergeometric functions and the other as combinatorial sums over partitions. We use our results to put forward conjectures on the joint moments of derivatives of $L$-functions with symplectic and orthogonal symmetry.
\end{abstract}

\maketitle

\section{Introduction}

Let $G(2N)\in\{Sp(2N),\ SO(2N),\ O^-(2N)\}$, where $Sp(2N)$ is the the group of $2N \times 2N$ unitary symplectic matrices and $SO(2N)$ and $O^-(2N)$ are the subsets of $2N \times 2N$ orthogonal matrices with determinant $+1$ and $-1$, respectively. Also, denote the characteristic polynomial of a matrix $A\in G(2N)$ by

\begin{equation*}
\Lambda_A(s)=\det \left( I-As \right).
\end{equation*}
In this paper we consider the joint moments

\begin{equation} \label{joint moments def}
\int_{G(2N)}  \left( \Lambda_A^{(n_1)}(1) \right)^{k_1} \left( \Lambda_A^{(n_2)}(1) \right)^{k_2} dA
\end{equation}
of the $n_1$-th and $n_2$-th derivatives of the characteristic polynomials, where $n_1, n_2$ are non-negative integers and $dA$ denotes the Haar measure on the relevant matrix ensemble. Using techniques developed in \cite{ABPRW14, CRS06, KW23a}, we obtain asymptotic formulae for (\ref{joint moments def}) for each $G(2N)$ and for all non-negative integers $k_1, k_2$. Our main results give two explicit expressions for the leading order coefficients for each of the matrix ensembles under consideration and are detailed in section \ref{Main results}.

The problem we study here is part of a general problem to obtain exact formulae for the complex moments of the derivatives of characteristic polynomials. A key motivation is the link between random matrix theory and the study of families of $L$-functions and their value distribution in analytic number theory. Specifically, one can use formulae obtained for characteristic polynomials of the various matrix ensembles to predict formulae for the corresponding quantities for $L$-functions with the same symmetry type. The complex moments of the derivatives of characteristic polynomials and $L$-functions can then be used to infer information on the zeros of the derivatives through Jensen's formula. For results on the radial distribution of the zeros of the derivative of characteristic polynomials and on the horizontal distribution of the zeros of the derivative of the Riemann zeta function, see, for example, \cite{ DFFHMP10, Mezzadri03} and \cite{Sou98, Zhang}, respectively.

Additionally on the number theory side, the order of vanishing of an $L$-function at the central point, which is controlled by the derivatives of the $L$-function, is widely believed to contain deep arithmetic and geometric information. The Birch and Swinnerton-Dyer Conjecture for example, famously states that the order of vanishing of an $L$-function attached to an elliptic curve over $\mathbb{Q}$ is equal to the rank of the curve.

For the ensemble of random unitary matrices $U(N)$, Conrey, Rubinstein and Snaith \cite{CRS06} proved that for integer $k \geq 1$,

\begin{equation*}
\int_{U(N)} |\Lambda_A'(1)|^{2k} dA \sim c_k N^{k^2+2k},
\end{equation*}
where

\begin{equation*}
c_k=(-1)^{k(k+1)/2} \sum_{h=0}^k \binom{k}{h} \left( \frac{d}{dx} \right)^{k+h} \left( e^{-x} x^{-k^2/2} \det_{k \times k} \left( I_{i+j-1} (2 \sqrt{x}) \right) \right),
\end{equation*}
with $I_n(x)$ denoting the modified Bessel function of the first kind. Here, and throughout the paper, the indices $i$ and $j$ of the matrix in the determinant range from 1 to $k$. Also proven in \cite{CRS06} is a similar asymptotic formula for the $2k$-th moment of the derivative of an analogue of Hardy's $Z$-function. As an application, the authors use their results to make a conjecture for the moments of the derivative of the Riemann zeta function and of the $Z$-function. Forrester and Witte \cite{FW06} have given alternate expressions for the leading order coefficients obtained in \cite{CRS06} in terms of solutions to Painlevé III differential equations.

Concerning the joint moments over the unitary ensemble, one is interested in the quantity

\begin{equation} \label{unitary joint moments}
\int_{U(N)} |\Lambda_A^{(n_1)}(1)|^{2M} |\Lambda_A^{(n_2)}(1)|^{2M-2k} dA.
\end{equation}
When $n_1=1$ and $n_2=0$, Hughes \cite{Hughes} was able to show that the limit of (\ref{unitary joint moments})$/N^{k^2+2M}$ as $N\to\infty$ exists when $k$ and $M$ are integers. This result was extended to all suitable, real $k$ and $M$ by Assiotis, Keating and
Warren in \cite{AKW22}. In \cite{BCPRS19}, Bailey et al. obtain an asymptotic formula for (\ref{unitary joint moments}) when $k \geq M$ are both non-negative integers. 

In the case of general $n_1, n_2$, Barhoumi-Andr{\'e}ani \cite{BA20} gave an asymptotic formula for (\ref{unitary joint moments}) for integer $k$ and $M$ with $k \geq M$ and $k \geq 2$ where the leading order coefficient is given in the form of a certain $(k-1)$-fold real integral. Recently, Keating and Wei \cite{KW23a} have obtained asymptotic formulae for (\ref{unitary joint moments}) and for the joint moments of the $n_1$-th and $n_2$-th of the analogue of Hardy's $Z$-function for all integers $k \geq M \geq 0$. They give two explicit expressions for the leading order coefficients, one in terms of derivatives of determinants involving the modified Bessel function and the other as combinatorial sums involving partitions. They also use their results to motivate conjectures for the joint moments of the $n_1$-th and $n_2$-th derivatives of the Riemann zeta function and of the $Z$-function. The conjectures made in \cite{KW23a} are shown to agree with the known results of \cite{Hall99, Hall04, Ingham28}. In \cite{KW23b}, Keating and Wei further explore the structure and properties of their leading order coefficients. They establish recursive relations that the coefficients satisfy and also build a connection to a solution of the $\sigma$-Painlev{\'e} III$'$ equation.

Turning to the symplectic and orthogonal matrix ensembles, Altu\u{g} et al. \cite{ABPRW14} considered the moments of the $m$-th derivative

\begin{equation*}
M_k(G(2N), m):=\int_{G(2N)} \left( \Lambda_A^{(m)}(1) \right)^k dA.
\end{equation*}
Extending the results of \cite{CRS06} to these ensembles, they prove asymptotic formulae for $M_k(G(2N), m)$ as $N\to\infty$ for integer $k \geq 1$ when $G(2N)=Sp(2N)$ or $G(N)=SO(2N)$ and $m=2$, and when $G(2N)=O^-(2N)$ with $m=3$. One considers the second derivative rather than the first in the case of $Sp(2N)$ and $SO(2N)$ since $\Lambda_A'(1)$ can be expressed simply in terms of $\Lambda_A(1)$. Thus, the moments of the first derivative can be computed using the result of Keating and Snaith \cite{KS00} on the moments of $\Lambda_A(1)$. If $A\in O^-(2N)$, then $\Lambda_A(1)=0$ and $\Lambda_A''(1)$ has as simple expression in terms of $\Lambda_A'(1)$ Hence, in this case, it is the moments of $\Lambda_A'''(1)$ that are of interest.

The leading order coefficients obtained in \cite{ABPRW14} are given in terms of derivatives of determinants involving hypergeometric functions. These determinants are shown to satisfy a differential recurrence relation similar to a Toda lattice equation connected to $\tau$-function theory in the study of Painlev{\'e} differential equations. An interesting question put forward in \cite{ABPRW14} is whether there is a differential equation in the symplectic and orthogonal cases which plays a part analogous to Painlev{\'e} III in the unitary setting. Gharakhloo and Witte \cite{GW23} have made promising progress in this direction in their study of $2j-k$ and $j-2k$ bi-orthogonal polynomial systems on the unit circle.

The authors of \cite{ABPRW14} also use their results to make conjectures for the asymptotic behaviour of the moments of derivatives at the central point of $L$-functions with symplectic or orthogonal symmetry. After stating our results in section \ref{Main results}, we will extend these to give general conjectures for the joint moments of the derivatives of $L$-functions with these symmetry types.

\subsection{Notation}

Recall that an $N \times N$ matrix $A$ is said to be unitary if $A A^*=I$, where $A^*$ is the conjugate transpose of $A$. The unitary symplectic group $Sp(2N)$ is the subgroup of $2N \times 2N$ unitary matrices $A$ which satisfy $A^T \Omega A=\Omega$, where

\begin{equation*}
\Omega=\begin{pmatrix} 0 & I \\ -I & 0 \end{pmatrix},
\end{equation*}
with $I$ the $N \times N$ identity matrix. The special orthogonal group $SO(2N)$ and $O^-(2N)$ are the subsets of orthogonal $2N \times 2N$ matrices with determinant $+1$ and $-1$, respectively. Each of these matrix ensembles is endowed with the normalised Haar measure $dA$.

We will set $\i^2=-1$ and the variable $i$ will only be used as an index. For real numbers $x$, we use $[x]$ to denote the greatest integer less than or equal to $x$. We write $S_k$ for the set of permutations on $\{1, 2, \dots, k\}.$
The multinomial coefficient is defined as

\begin{equation*}
\binom{n}{l_1, \dots, l_k}=\frac{n!}{l_1! \cdots l_k!},
\end{equation*}
for integers $n$ and $l_1, \dots, l_k$ with $l_1+\cdots+l_k=n$. Also, for integer $n$, whenever we write $l_1+\cdots+l_k=n$ or $l_1+\cdots+l_k \leq n$ this means that $l_i \geq 0$ are taken to be integers. 

For any $w=(w_1,\dots, w_k)\in\mathbb{C}^k$, the Vandermonde determinant is denoted by

\begin{equation*}
\Delta(w):=\det_{k\times k} (w_i^{j-1})=\prod_{1\leq i<j\leq k} (w_i-w_j),
\end{equation*}
and we write $w^2=(w_i^2)_{1\leq i\leq k}$. We will also make use of Vandermonde determinants of differential operators, written as

\begin{equation*}
\Delta \left( \frac{d}{dx} \right):=\det_{k\times k} \left( \frac{d^{j-1}}{dx_i^{j-1}} \right)=\prod_{1\leq i<j\leq k} \left( \frac{d}{dx_i}-\frac{d}{dx_j} \right).
\end{equation*}

Lastly, for $u\in\mathbb{C}$ and $m\in\mathbb{Z}$, we let

\begin{align} \label{hypergeometric def}
g_m (u) & :=\frac{1}{2\pi\i} \oint_{|w|=1} \frac{e^{w+u/w^2}}{w^{m+1}} dw \nonumber \\
& =\frac{1}{\Gamma(m+1)} {}_0 F_2 \left( \frac{m}{2}+1,\frac{m+1}{2};\frac{u}{4} \right).
\end{align}
These hypergeometric functions will play the role that the modified Bessel function plays in the unitary case. For negative $m$, one should interpret the above expression as the limit.

\section{Main results} \label{Main results}

We now state our main results. Our first two theorems give an asymptotic formula for the joint moments of derivatives of characteristic polynomials of matrices over $Sp(2N)$.

\begin{theorem} \label{Sp theorem 1}
Let $0\leq n_1 \leq n_2$ be integers and let $k_1, k_2$ be non-negative integers, not both $0$. Set $k=k_1+k_2$. Then, we have

\begin{equation*}
\int_{Sp(2N)} \left( \Lambda_A^{(n_1)}(1) \right)^{k_1} \left( \Lambda_A^{(n_2)}(1) \right)^{k_2} dA=b_{k_1, k_2}^{Sp}(n_1, n_2) \cdot (2N)^{k(k+1)/2+k_1 n_1+k_2 n_2} \left(1+O(N^{-1}) \right),
\end{equation*}
where

\begin{align*}
b_{k_1, k_2}^{Sp}(n_1, n_2) & =\frac{(-1)^{k_1 n_1+k_2 n_2}}{2^{k(k+1)/2+k_1 n_1+k_2 n_2}} \sum_{u_1+\cdots+u_P=k_1} \binom{k_1}{u_1\dots, u_P} \frac{(n_1!)^{k_1}}{\prod_{i=1}^P (\textbf{a}_i!)^{u_i} (\prod_{j=1}^{[n_1/2]} j^{\sum_{i=1}^P u_i a_{i,j}})} \\
& \times \sum_{v_1+\cdots+v_Q=k_2} \binom{k_2}{v_1\dots, v_Q} \frac{(n_2!)^{k_2}}{\prod_{i=1}^Q (\textbf{b}_i!)^{v_i} (\prod_{j=1}^{[n_2/2]} j^{\sum_{i=1}^Q v_i b_{i,j}})} \\
& \times \sum_{\substack{\sum_{i=1}^k m_{s, i}=m_s \\ s=2, \dots, n}} \left( \prod_{s=2}^n \binom{m_s}{m_{s, 1}, \dots, m_{s, k}} \right) \left( \frac{d}{du} \right)^{m_1} \left. \det_{k\times k} \left( g_{2i-j+2 \sum_{s=2}^n s m_{s, i}}(u) \right) \right|_{u=0}.
\end{align*}
Here, for $i=1, \dots, P$, we define ${\bf a}_i=(a_{i, 0}, a_{i,1}, \dots, a_{i, [n_1/2]})$ to be the $P$ distinct tuples of integers satisfying

\begin{equation*}
a_{i, j} \geq 0 \ \textup{and} \ a_{i, 0}+2 \sum_{j=1}^{[n_1/2]} j a_{i, j}=n_1.
\end{equation*}
Similarly, the ${\bf b}_i=(b_{i, 0}, b_{i, 1}, \dots, b_{i, [n_2/2]})$ for $i=1, \dots, Q$ are defined to be the $Q$ tuples of integers satisfying

\begin{equation*}
b_{i, j} \geq 0 \ \textup{and} \ b_{i, 0}+2 \sum_{j=1}^{[n_2/2]} j b_{i, j}=n_2.
\end{equation*}
Also, we define ${\bf a}_i!:=\prod_{j=1}^{[n_1/2]} a_{i,j}!$ and ${\bf b}_i!:=\prod_{j=1}^{[n_2/2]} b_{i,j}!$. Finally, $m_j:=\sum_{i=1}^P u_i a_{i, j}+\sum_{i=1}^Q v_i b_{i,j}$ for $j=1, \dots, [n_1/2]$ and $m_j:=\sum_{i=1}^Q v_i b_{i,j}$ for $j=[n_1/2]+1, \dots, [n_2/2]$.
\end{theorem}

\begin{theorem} \label{Sp theorem 2}
Let $0\leq n_1 \leq n_2$ be integers and let $k_1, k_2$ be non-negative integers, not both $0$. Set $k=k_1+k_2$. Then, we have

\begin{equation*}
\int_{Sp(2N)} \left( \Lambda_A^{(n_1)}(1) \right)^{k_1} \left( \Lambda_A^{(n_2)}(1) \right)^{k_2} dA=b_{k_1, k_2}^{Sp}(n_1, n_2) \cdot (2N)^{k(k+1)/2+k_1 n_1+k_2 n_2} \left(1+O(N^{-1}) \right),
\end{equation*}
where

\begin{align*}
b_{k_1, k_2}^{Sp}(n_1, n_2) & =\frac{(-1)^{k(k-1)/2+k_1 n_1+k_2 n_2}}{2^{k(k+1)/2+k_1 n_1+k_2 n_2}} (n_1!)^{k_1} (n_2!)^{k_2} \\
&\quad \times \sum_{\substack{2 \sum_{j=1}^k l_{i, j} \leq n_1 \\ i=1, \dots, k_1}} \sum_{\substack{2 \sum_{j=1}^k m_{i, j} \leq n_2 \\ i=1, \dots, k_2}} \left( \prod_{i=1}^{k_1} \frac{1}{(n_1-2 \sum_{j=1}^k l_{i ,j})!} \right) \left( \prod_{i=1}^{k_2} \frac{1}{(n_2-2 \sum_{j=1}^k m_{i, j})!} \right) \\
&\quad \times \prod_{j=1}^k \frac{1}{(2k+V_j-2j+1)!} \prod_{1 \leq i<j \leq k} (V_j-V_i-2j+2i).
\end{align*}
Here $V_j:=2\sum_{i=1}^{k_1} l_{i, j}+2\sum_{i=1}^{k_2} m_{i, j}$ for $j=1, \dots, k$.
\end{theorem}

Our next two theorems give an asymptotic formula for the joint moments over $SO(2N)$.

\begin{theorem} \label{SO theorem 1}
With notation as in Theorem \ref{Sp theorem 1}, we have

\begin{equation*}
\int_{SO(2N)} \left( \Lambda_A^{(n_1)}(1) \right)^{k_1} \left( \Lambda_A^{(n_2)}(1) \right)^{k_2} dA=b_{k_1, k_2}^{SO}(n_1, n_2) \cdot (2N)^{k(k-1)/2+k_1 n_1+k_2 n_2} \left( 1+O(N^{-1}) \right),
\end{equation*}
where

\begin{align*}
b_{k_1, k_2}^{SO}(n_1, n_2) & =\frac{1}{2^{k(k-3)/2+k_1 n_1+k_2 n_2}} \sum_{u_1+\cdots+u_P=k_1} \binom{k_1}{u_1\dots, u_P} \frac{(n_1!)^{k_1}}{\prod_{i=1}^P (\textbf{a}_i!)^{u_i} (\prod_{j=1}^{[n_1/2]} j^{\sum_{i=1}^P u_i a_{i,j}})} \\
& \times \sum_{v_1+\cdots+v_Q=k_2} \binom{k_2}{v_1\dots, v_Q} \frac{(n_2!)^{k_2}}{\prod_{i=1}^Q (\textbf{b}_i!)^{v_i} (\prod_{j=1}^{[n_2/2]} j^{\sum_{i=1}^Q v_i b_{i,j}})} \\
& \times \sum_{\substack{\sum_{i=1}^k m_{s, i}=m_s \\ s=2, \dots, n}} \left( \prod_{s=2}^n \binom{m_s}{m_{s, 1}, \dots, m_{s, k}} \right) \left( \frac{d}{du} \right)^{m_1} \left. \det_{k\times k} \left( g_{2i-j-1+2 \sum_{s=2}^n s m_{s, i}}(u) \right) \right|_{u=0}.
\end{align*}
\end{theorem}

\begin{theorem} \label{SO theorem 2}
With notation as in Theorem \ref{Sp theorem 2}, we have

\begin{equation*}
\int_{SO(2N)} \left( \Lambda_A^{(n_1)}(1) \right)^{k_1} \left( \Lambda_A^{(n_2)}(1) \right)^{k_2} dA=b_{k_1, k_2}^{SO}(n_1, n_2) \cdot (2N)^{k(k-1)/2+k_1 n_1+k_2 n_2} \left( 1+O(N^{-1}) \right),
\end{equation*}
where

\begin{align*}
b_{k_1, k_2}^{SO}(n_1, n_2) & =\frac{(-1)^{k(k-1)/2}}{2^{k(k-3)/2+k_1 n_1+k_2 n_2}} (n_1!)^{k_1} (n_2!)^{k_2} \\
&\quad \times \sum_{\substack{2 \sum_{j=1}^k l_{i, j} \leq n_1 \\ i=1, \dots, k_1}} \sum_{\substack{2 \sum_{j=1}^k m_{i, j} \leq n_2 \\ i=1, \dots, k_2}} \left( \prod_{i=1}^{k_1} \frac{1}{(n_1-2 \sum_{j=1}^k l_{i ,j})!} \right) \left( \prod_{i=1}^{k_2} \frac{1}{(n_2-2 \sum_{j=1}^k m_{i, j})!} \right) \\
&\quad \times \prod_{j=1}^k \frac{1}{(2k+V_j-2j)!} \prod_{1 \leq i<j \leq k} (V_j-V_i-2j+2i).
\end{align*}

\end{theorem}

Our final theorem gives an asymptotic formula for the joint moments over $O^-(2N)$ with the leading order coefficient expressed in terms of $b_{k_1, k_2}^{Sp}(n_1, n_2)$.

\begin{theorem} \label{O^- theorem}
Let $1 \leq n_1 \leq n_2$ be integers and let $k_1, k_2$ be non-negative integers, not both 0. Set $k=k_1+k_2$. Then, we have

\begin{equation*}
\int_{O^-(2N)} \left( \Lambda_A^{(n_1)}(1) \right)^{k_1} \left( \Lambda_A^{(n_2)}(1) \right)^{k_2} dA=b_{k_1, k_2}^{O^-}(n_1, n_2) \cdot (2N)^{k(k+1)/2+k_1 (n_1-1)+k_2 (n_2-1)} \left( 1+O(N^{-1}) \right),
\end{equation*}
where

\begin{equation*}
b_{k_1, k_2}^{O^-}(n_1, n_2)=(-1)^{k_1 (n_1-1)+k_2 (n_2-1)} \, 2^k n_1^{k_1} n_2^{k_2} \, b_{k_1, k_2}^{Sp}(n_1-1, n_2-1),
\end{equation*}
with $b_{k_1, k_2}^{Sp}(n_1, n_2)$ as defined in Theorems \ref{Sp theorem 1} and \ref{Sp theorem 2}.
\end{theorem}

Our Theorems \ref{Sp theorem 1}, \ref{SO theorem 1} exhibit the same structure as the asymptotic formulae obtained in \cite{ABPRW14}. Namely, the leading order coefficients are expressed in terms of derivatives of determinants of the hypergeometric functions $g_m(u)$. As mentioned in the introduction, these determinants were shown to satisfy a differential recurrence relation \cite[Theorem 1.5]{ABPRW14} which allows the leading order coefficients to be computed much more quickly as $k_1, k_2$ get large. However, similarly to the unitary case considered in \cite{KW23a}, the formulae given for the leading order coefficients in Theorems \ref{Sp theorem 1} and \ref{SO theorem 1} may not be computationally efficient when $n_1, n_2$ are large since one has to compute the tuples ${\bf a}_i$ and ${\bf b}_i$. Aside from giving an alternate expression for the leading order coefficients, the advantage of Theorems \ref{Sp theorem 2} and \ref{SO theorem 2} is that the formulae are more computationally effective when $n_1, n_2$ are large and $k_1, k_2$ are small.

Using the standard random matrix philosophy allows us to make conjectures based on our results for the joint moments of derivatives of $L$-functions with symmetry type $Sp, \, SO$ or $O^-$ in the sense of \cite{KS99}. We give an example conjecture for the family of quadratic Dirichlet $L$-functions at $s=1/2$ below. This is an example of a family with symplectic symmetry so we use our results for $Sp(2N)$ as a model.

\begin{conjecture}
Let $\mathcal{D}(X)=\{ d \ \textup{a fundamental discriminant}: |d|< X \}$, and let $L(s,\chi_d)$ be the Dirichlet $L$-function attached to the quadratic character $\chi_d$. Then, for $0 \leq n_1 \leq n_2$ and $k_1, k_2 \geq 0$ integers with $k_1, k_2$ not both 0, we have that as $X \to \infty$,

\begin{equation*}
\frac{1}{|\mathcal{D}(X)|} \sum_{d \in \mathcal{D}(X)} L^{(n_1)}(1/2,\chi_d)^{k_1} L^{(n_2)}(1/2,\chi_d)^{k_2} \sim a_k \cdot b_{k_1, k_2}^{Sp}(n_1, n_2) \cdot (\log X)^{k(k+1)/2+k_1 n_1+k_2 n_2},
\end{equation*}
where $k=k_1+k_2$ and $b_{k_1, k_2}^{Sp}(n_1, n_2)$ is the random matrix theory coefficient defined in Theorems \ref{Sp theorem 1} and \ref{Sp theorem 2}. Also,

\begin{equation*}
a_k=\prod_{p \ \textup{prime}} \frac{(1-1/p)^{k(k+1)/2}}{1+1/p} \left( \frac{(1-1/\sqrt{p})^{-k}+(1+1/\sqrt{p})^{-k}}{2}+\frac{1}{p} \right),
\end{equation*}
is an arithmetic constant depending on the family of $L$-functions. We note that $a_k$ is the same coefficient appearing in the conjectures for the moments of $L(1/2,\chi_d)$, see \cite{CFKRS05, KS00}.
\end{conjecture}

To the best of our knowledge, there are no results on the moments of derivatives of these quadratic Dirichlet $L$-functions over number fields. However, the moments of derivatives of $L$-functions over function fields have been investigated. For this discussion, we let $\mathbb{F}_q$ denote a finite field with $q$ elements and let $\mathcal{H}_{2g+1}$ be the subset of square-free, monic polynomials of degree $2g+1$ in the polynomial ring $\mathbb{F}_q[x]$. For each $D \in \mathcal{H}_{2g+1}$, we have a Dirichlet $L$-function $L(s,\chi_D)$ attached to the quadratic character $\chi_D$ with conductor $|D|:=q^{2g+1}$. This family of $L$-functions also exhibits symplectic symmetry. Andrade and Rajagopal \cite{AR16} and subsequently Andrade and Jung \cite{AJ21} studied the mean values of $L^{(n)}(1/2,\chi_D)$ and an asymptotic formula for the first moment of $L^{(n)}(1/2,\chi_D)$ is given in \cite{AJ21} which implies that for any positive integer $n$ as $g \to \infty$,

\begin{equation} \label{asymptotic formula 1}
\frac{1}{|\mathcal{H}_{2g+1}|} \sum_{D \in \mathcal{H}_{2g+1}} L^{(n)}(1/2,\chi_D) \sim \frac{(-1)^n}{2(n+1)} \mathcal{A}(1) \cdot (2g+1)^{n+1}.
\end{equation}
Here, $\mathcal{A}(1)$ is given as an Euler product over the monic, irreducible polynomials in $\mathbb{F}_q[x]$ and is the same arithmetic factor that appears in the asymptotic formula for the first moment of $L(1/2,\chi_D)$. The following proposition shows that (\ref{asymptotic formula 1}) is indeed the asymptotic formula predicted by our conjecture in this case.

\begin{proposition} \label{b(0,n) prop}
For $n \geq 1$ an integer, we have that

\begin{equation*}
b_{0,1}^{Sp}(0,n)=\frac{(-1)^n}{2(n+1)}.
\end{equation*}
Also, for $n \geq 1$, we have

\begin{equation*}
b_{0,1}^{SO}(0,n)=1.
\end{equation*}
\end{proposition}

A mixed second moment involving a second derivative of these quadratic Dirichlet $L$-functions over function fields was considered by Djankovi{\'c} and \DJ oki{\'c} \cite{DD21}. In particular, Theorem 1.1 and Remark 1.2 in \cite{DD21}, along with Florea's asymptotic formula for the second moment of $L(1/2,\chi_D)$ in \cite{Florea}, imply that

\begin{equation*}
\frac{1}{|\mathcal{H}_{2g+1}|} \sum_{D \in \mathcal{H}_{2g+1}} \frac{L(1/2,\chi_D) L''(1/2,\chi_D)}{\log^2 q} \sim \frac{(1-q^{-1})}{80} \mathcal{B}(1/q) \cdot (2g+1)^5.
\end{equation*}
Similarly to the previous asymptotic formula, $\mathcal{B}(1/q)$ is given as an Euler product over monic, irreducible polynomials and it agrees with the arithmetic factor appearing in the asymptotic formula for the second moment of $L(1/2,\chi_D)$. This result is therefore consistent with the prediction of our conjecture as we compute that $b_{1,1}^{Sp}(0,2)=1/80$.

One can naturally use our Theorems \ref{Sp theorem 1}-\ref{O^- theorem} to make analogous conjectures for the joint moments of derivatives at the central point for any family of $L$-functions with symplectic or orthogonal symmetry.

\section{Preliminaries}

We begin with the following two lemmas concerning Vandermonde determinants of differential operators. The first is quoted from \cite[Lemma 2.8]{ABPRW14} and follows from the definition.

\begin{lemma} \label{Vandermonde 1}
Let $f_1(x),\dots, f_k(x)$ be $k-1$ times differentiable. Then

\begin{equation*}
\Delta \left( \frac{d}{dx} \right) \prod_{i=1}^k f_i(x_i)=\det_{k \times k} \left( f_i^{(j-1)}(x_i) \right).
\end{equation*}
\end{lemma}

\begin{lemma} \label{Vandermonde 2}
Let $f_1(x,y), \dots, f_k(x,y)$ be $k-1$ times differentiable in $x$ and $y$. Then

\begin{equation*}
\Delta \left( \frac{d}{dx} \right) \Delta \left( \frac{d}{dy} \right) \left. \prod_{i=1}^k f_i(x_i, y_i) \right|_{\substack{x_1=\cdots=x_k=X, \\ y_1=\cdots=y_k=Y}}=\sum_{\mu\in S_k} \det_{k \times k} \left( \frac{d^{i+j-2}}{dX^{i-1} dY^{j-1}} f_{\mu(i)}(X,Y) \right).
\end{equation*}
In particular, when $f_1=\cdots=f_k=f$, we have

\begin{equation*}
\Delta \left( \frac{d}{dx} \right) \Delta \left( \frac{d}{dy} \right) \left. \prod_{i=1}^k f(x_i, y_i) \right|_{\substack{x_1=\cdots=x_k=X, \\ y_1=\cdots=y_k=Y}}=k! \det_{k\times k} \left( \frac{d^{i+j-2}}{dX^{i-1} dY^{j-1}} f(X,Y) \right).
\end{equation*}
\end{lemma}

\begin{proof}
The case when $f_1=\cdots=f_k=f$ is the result of \cite[Lemma 2.9]{ABPRW14} and the proof of the general case follows the same lines. By Lemma \ref{Vandermonde 1}, we have

\begin{equation*}
\Delta \left( \frac{d}{dx} \right) \prod_{i=1}^k f_i(x_i, y_i)=\det_{k \times k} \left( \frac{d^{j-1}}{dx_i^{j-1}} f_i(x_i, y_i) \right)=\sum_{\mu\in S_k} \textrm{sign}(\mu) \prod_{i=1}^k \frac{d^{\mu(i)-1}}{dx_i^{\mu(i)-1}} f_i(x_i, y_i).
\end{equation*}
Then, by Lemma \ref{Vandermonde 1} again, we have that

\begin{align*}
\Delta \left( \frac{d}{dx} \right) \Delta \left( \frac{d}{dy} \right) \left. \prod_{i=1}^k f_i(x_i, y_i)  \right|_{\substack{x_i=X, \\ y_i=Y}} &=\sum_{\mu\in S_k} \textrm{sign}(\mu) \, \Delta \left( \frac{d}{dy} \right) \left. \prod_{i=1}^k \frac{d^{\mu(i)-1}}{dx_i^{\mu(i)-1}} f_i(x_i, y_i)  \right|_{\substack{x_i=X, \\ y_i=Y}} \\
&=\sum_{\mu\in S_k} \textrm{sign}(\mu) \left. \det_{k \times k} \left( \frac{d^{\mu(i)+j-2}}{dx_i^{\mu(i)-1} dy_i^{j-1}} f_i(x_i, y_i) \right)  \right|_{\substack{x_i=X, \\ y_i=Y}} \\
&=\sum_{\mu\in S_k} \textrm{sign}(\mu) \det_{k \times k} \left( \frac{d^{\mu(i)+j-2}}{dX^{\mu(i)-1} dY^{j-1}} f_i(X,Y) \right) \\
&=\sum_{\mu\in S_k} \det_{k \times k} \left( \frac{d^{i+j-2}}{dX^{i-1} dY^{j-1}} f_{\mu(i)}(X,Y) \right),
\end{align*}
where we have interchanged the rows of the matrix to obtain the final line.
\end{proof}

We next express a certain contour integral in terms of the hypergeometric functions $g_m(u)$.

\begin{lemma} \label{hypergeometric integral}
Let $k\in\mathbb{Z}$ and let $n\geq 1$ be an integer. Then, for complex numbers $u_1,\dots, u_n$, we have

\begin{equation*}
\frac{1}{2\pi\i} \oint_{|w|=1} \exp \left( w+\sum_{j=1}^n \frac{u_j}{w^{2j}} \right) \frac{dw}{w^{k+1}}=\sum_{m_2,\dots, m_n=0}^{\infty} \left( \prod_{j=2}^n \frac{u_j^{m_j}}{m_j!} \right) g_{k+2\sum_{j=2}^n jm_j}(u_1),
\end{equation*}
where $g_m(u)$ is the hypergeometric function defined in \textup{(\ref{hypergeometric def})}.
\end{lemma}

\begin{proof}
We compute the integral by determining the coefficient of $w^k$ in the exponential factor of the integrand. So, let $a_n(k)$ be the coefficient of $w^k$ in $\exp(w+\sum_{j=1}^n \frac{u_j}{w^{2j}})$. Then,

\begin{align*}
\exp \left (w+\sum_{j=1}^n \frac{u_j}{w^{2j}} \right) & =\exp \left( \frac{u_n}{w^{2n}} \right) \exp \left (w+\sum_{j=1}^{n-1} \frac{u_j}{w^{2j}} \right) \\
& =\left( \sum_{m=0}^{\infty} \frac{u_n^m}{m!} w^{-2nm} \right) \left( \sum_{m=-\infty}^{\infty} a_{n-1}(m) w^m \right).
\end{align*}
From this it follows that

\begin{align*}
a_n(k) & =\sum_{m_n=0}^{\infty} \frac{u_n^{m_n}}{m_n!} a_{n-1}(k+2nm_n) \\
& = \sum_{m_2,\dots, m_n=0}^{\infty} \left( \prod_{j=2}^n \frac{u_j^{m_j}}{m_j!} \right) a_1(k+2\sum_{j=2}^n jm_j).
\end{align*}
We then see that by definition, $a_1(k+2\sum_{j=2}^n jm_j)=g_{k+2\sum_{j=2}^n jm_j}(u_1)$ and hence

\begin{equation*}
a_n(k)=\sum_{m_2,\dots, m_n=0}^{\infty} \left( \prod_{j=2}^n \frac{u_j^{m_j}}{m_j!} \right) g_{k+2\sum_{j=2}^n jm_j}(u_1),
\end{equation*}
as required.
\end{proof}

The next lemma allows us to take higher order derivatives of determinants of functions.

\begin{lemma}[Lemma 13 in \cite{KW23a}] \label{determinant derivative}
Let $s\geq 0$, $k\geq 1$ be integers and $a_{i,j}(x)$ be $s$-th differentiable functions of $x$. Then

\begin{equation*}
\left( \frac{d}{dx} \right)^s \det_{k\times k} \left( a_{i,j}(x) \right)=\sum_{l_1+\cdots+l_k=s} \binom{s}{l_1, \dots, l_k} \det_{k\times k} \left( a_{i,j}^{(l_i)}(x) \right),
\end{equation*}
where $a_{i,j}^{(l_i)}(x)$ means that we take the $l_i$-th derivative of $a_{i,j}(x)$.
\end{lemma}

Now, the shifted moments of the characteristic polynomials are defined as
\begin{equation*}
I(G(2N); z_1, \dots, z_k):=\int_{G(2N)} \Lambda_A(z_1) \cdots \Lambda_A(z_k) \, dA.
\end{equation*}
These shifted moments have been computed by Conrey et al. \cite{CFKRS03} and can be expressed in the form of a multiple contour integral. We will use the following approximate versions of their formulae which follow easily from the results of \cite{CFKRS03} and the fact that $(1-e^{-x})^{-1}=x^{-1}+O(1)$ for small $x$.

\begin{lemma}[Corollary 2.4 in \cite{ABPRW14}] \label{Sp shifted moments}
Let $\alpha_1, \dots, \alpha_k$ be complex numbers such that $|\alpha_j|\ll 1/N$ for $j=1,2, \dots, k$. Then

\begin{align*}
& I(Sp(2N); e^{-\alpha_1}, \dots, e^{-\alpha_k}) \\
&\qquad =\frac{(-1)^{k(k-1)/2}}{(2\pi\i)^k k!} \oint \cdots \oint_{|w_i|=1} \frac{\Delta(w) \Delta(w^2) \, e^{N \sum_{i=1}^k (w_i-\alpha_i)}}{\prod_{1 \leq i,j \leq k} (w_i^2-\alpha_j^2)} \prod_{i=1}^k dw_i \left( 1+O(N^{-1}) \right).
\end{align*}
\end{lemma}

\begin{lemma}[Corollary 2.5 in \cite{ABPRW14}] \label{SO shifted moments}
Let $\alpha_1, \dots, \alpha_k$ be complex numbers such that $|\alpha_j|\ll 1/N$ for $j=1,2, \dots, k$. Then

\begin{align*}
& I(SO(2N); e^{-\alpha_1}, \dots, e^{-\alpha_k}) \\
&\qquad =\frac{(-1)^{k(k-1)/2} 2^k}{(2\pi\i)^k k!} \oint \cdots \oint_{|w_i|=1} \frac{\Delta(w) \Delta(w^2) (\prod_{i=1}^k w_i) \, e^{N \sum_{i=1}^k (w_i+\alpha_i)}}{\prod_{1 \leq i,j \leq k} (w_i^2-\alpha_j^2)} \prod_{i=1}^k dw_i \left( 1+O(N^{-1}) \right).
\end{align*}
\end{lemma}

\begin{lemma}[Corollary 2.6 in \cite{ABPRW14}] \label{O^- shifted moments}
Let $\alpha_1, \dots, \alpha_k$ be complex numbers such that $|\alpha_j|\ll 1/N$ for $j=1,2, \dots, k$. Then

\begin{align*}
& I(O^-(2N); e^{-\alpha_1}, \dots, e^{-\alpha_k}) \\
&\qquad =\frac{(-1)^{k(k-1)/2} 2^k}{(2\pi\i)^k k!} \oint \cdots \oint_{|w_i|=1} \frac{\Delta(w) \Delta(w^2) (\prod_{i=1}^k \alpha_i) \, e^{N \sum_{i=1}^k (w_i+\alpha_i)}}{\prod_{1 \leq i,j \leq k} (w_i^2-\alpha_j^2)} \prod_{i=1}^k dw_i \left( 1+O(N^{-1}) \right).
\end{align*}
\end{lemma}

Below we give two expressions for the derivatives of these contour integral expressions for shifted moments with respect to the shifts $\alpha_j$.

\begin{lemma} \label{1st derivative}
Let $n \geq 0$ and $k \geq 1$ be integers. Then

\begin{equation*}
\frac{d^n}{d\alpha^n} \left. \frac{e^{-N\alpha}}{\prod_{i=1}^k (w_i^2-\alpha^2)}\right|_{\alpha=0}=\left( \prod_{i=1}^k \frac{1}{w_i^2} \right) \sum_{m=0}^n \binom{n}{m} (-N)^{n-m} \, m! \sum_{\substack{l_1+\cdots+l_k=m \\ l_j \ \textup{even}}} \prod_{i=1}^k \frac{1}{w_i^{l_i}}.
\end{equation*}
\end{lemma}

\begin{proof}
This follows from the proof of \cite[Lemma 2.7]{ABPRW14} where we have corrected a typo.
\end{proof}

\begin{lemma} \label{2nd derivative}
Let $n \geq 0$ and $k \geq 1$ be integers. Then

\begin{equation*}
\frac{d^n}{d\alpha^n} \left. \frac{e^{-N\alpha}}{\prod_{i=1}^k (w_i^2-\alpha^2)}\right|_{\alpha=0}= \left(\prod_{i=1}^k \frac{1}{w_i^2}\right) \sum_{\substack{m_1+2m_2+\cdots+n m_n=n \\ m_3=m_5=\cdots=0}} \frac{n!}{m_1!\cdots m_n!} (-N)^{m_1} \prod_{j=1}^{[n/2]} \left(\frac{1}{j} \sum_{i=1}^k \frac{1}{w_i^{2j}}\right)^{m_{2j}}.
\end{equation*}


\end{lemma}

\begin{proof}
The proof is similar to that of \cite[Lemma 9]{KW23a}. First, we have

\begin{equation*}
\frac{d}{d\alpha} \frac{e^{-N\alpha}}{\prod_{i=1}^k (w_i^2-\alpha^2)}=\frac{e^{-N\alpha}}{\prod_{i=1}^k (w_i^2-\alpha^2)} f_1 (\alpha),
\end{equation*}
where

\begin{equation*}
f_1 (\alpha)=-N+2\alpha \sum_{i=1}^k \frac{1}{w_i^2-\alpha^2}.
\end{equation*}
We can then write

\begin{equation} \label{recursive diff}
\frac{d^n}{d\alpha^n} \frac{e^{-N\alpha}}{\prod_{i=1}^k (w_i^2-\alpha^2)}=\frac{e^{-N\alpha}}{\prod_{i=1}^k (w_i^2-\alpha^2)} f_n(\alpha),
\end{equation}
where $f_n (\alpha)$ is defined recursively by

\begin{equation*}
f_{n+1}(\alpha)=f_n(\alpha) f_1(\alpha)+f_n'(\alpha).
\end{equation*}
Now, let $g(\alpha)$ be a function such that $g'(\alpha)=f_1(\alpha)$. Then, we have that

\begin{equation*}
\frac{d^n}{d\alpha^n} e^{g(\alpha)}=e^{g(\alpha)} f_n(\alpha).
\end{equation*}
But, by Fa{\`a} di Bruno's formula, we also have that

\begin{align*}
\frac{d^n}{d\alpha^n} e^{g(\alpha)} & =e^{g(\alpha)} \sum_{m_1+2m_2+\cdots+n m_n=n} \frac{n!}{m_1! \cdots m_n!} \prod_{j=1}^n \left( \frac{g^{(j)}(\alpha)}{j!} \right)^{m_j} \\
& =e^{g(\alpha)} \sum_{m_1+2m_2+\cdots+n m_n=n} \frac{n!}{m_1! \cdots m_n!} \prod_{j=1}^n \left( \frac{f_1^{(j-1)}(\alpha)}{j!} \right)^{m_j}.
\end{align*}
Comparing the above two expressions for $(d/d\alpha)^n e^{g(\alpha)}$, we see that

\begin{equation*}
f_n(\alpha)=\sum_{m_1+2m_2+\cdots+n m_n=n} \frac{n!}{m_1!\cdots m_n!} \prod_{j=1}^n \left( \frac{f_1^{(j-1)}(\alpha)}{j!} \right)^{m_j}.
\end{equation*}
One can check that for $j\geq 1$, we have

\begin{align*}
f_1^{(j)}(0)=\begin{cases} 0 &\textrm{if} \quad j \quad \textrm{even}, \\ 2j! \sum_{i=1}^k w_i^{-(1+j)} &\textrm{if} \quad j \quad \text{odd}. \end{cases}
\end{align*}
Hence, we have that

\begin{equation*}
f_n(0)=\sum_{\substack{m_1+2m_2+\cdots+nm_n=n \\ m_3=m_5=\cdots=0}} \frac{n!}{m_1! \cdots m_n!} (-N)^{m_1} \prod_{j=1}^{[n/2]} \left( \frac{1}{j} \sum_{i=1}^k \frac{1}{w_i^{2j}} \right)^{m_{2j}}.
\end{equation*}
Evaluating (\ref{recursive diff}) at $\alpha=0$ using this expression for $f_n(0)$ yields the desired result.

\end{proof}

In the next two propositions we will compute the main contour integrals that we need to evaluate.

\begin{proposition} \label{1st integral prop}
Let $k\geq 1$ and $n\geq 1$ be integers. Also, let $(m_1,\dots,m_n)$ be a tuple of non-negative integers. Then, we have

\begin{align} \label{integral 1}
& \frac{1}{(2\pi\i)^k} \oint \cdots \oint_{|w_i|=1} \Delta(w) \Delta(w^2) \, e^{N \sum_{i=1}^k w_i} \prod_{j=1}^n \left( \sum_{i=1}^k \frac{1}{w_i^{2j}} \right)^{m_j} \prod_{i=1}^k \frac{dw_i}{w_i^{2k}} \nonumber \\
&=(-1)^{k(k-1)/2} k! N^{k(k+1)/2+2 \sum_{j=1}^n j m_j} \sum_{\substack{\sum_{i=1}^k m_{s, i}=m_s \\ s=2, \dots, n}} \left( \prod_{s=2}^n \binom{m_s}{m_{s, 1}, \dots, m_{s, k}} \right) \nonumber \\
&\qquad \times \left( \frac{d}{du} \right)^{m_1} \left. \det_{k\times k} \left( g_{2i-j+2 \sum_{s=2}^n s m_{s, i}}(u) \right) \right|_{u=0}.
\end{align}
Also, we have

\begin{align} \label{integral 2}
& \frac{1}{(2\pi\i)^k} \oint \cdots \oint_{|w_i|=1} \Delta(w) \Delta(w^2) \, e^{N \sum_{i=1}^k w_i} \prod_{j=1}^n \left( \sum_{i=1}^k \frac{1}{w_i^{2j}} \right)^{m_j} \prod_{i=1}^k \frac{dw_i}{w_i^{2k-1}} \nonumber \\
&  =(-1)^{k(k-1)/2} k! N^{k(k-1)/2+2 \sum_{j=1}^n j m_j} \sum_{\substack{\sum_{i=1}^k m_{s, i}=m_s \\ s=2, \dots, n}} \left( \prod_{s=2}^n \binom{m_s}{m_{s, 1}, \dots, m_{s, k}} \right) \nonumber \\
&\qquad \times \left( \frac{d}{du} \right)^{m_1} \left. \det_{k\times k} \left( g_{2i-j-1+2 \sum_{s=2}^n s m_{s, i}}(u) \right) \right|_{u=0}.
\end{align}
\end{proposition}

\begin{proof}
First, note that

\begin{equation*}
\Delta(w^2)=\Delta \left( \frac{d}{dX} \right) \left. \exp \left( \sum_{i=1}^k w_i^2 X_i \right) \right|_{X_i=0},
\end{equation*}
and

\begin{equation*}
\Delta(w) \, e^{N \sum_{i=1}^k w_i}=\Delta \left( \frac{d}{dY} \right) \left. \exp \left( \sum_{i=1}^k w_i Y_i \right) \right|_{Y_i=N}.
\end{equation*}
We may also write

\begin{equation*}
\prod_{j=1}^n \left( \sum_{i=1}^k \frac{1}{w_i^{2j}} \right)^{m_j}=\prod_{j=1}^n \left( \frac{d}{dt_j} \right)^{m_j} \left. \exp \left( \sum_{j=1}^n t_j \sum_{i=1}^k \frac{1}{w_i^{2j}} \right) \right|_{t_j=0}.
\end{equation*}
Then, we have that

\begin{align*}
& \frac{1}{(2\pi\i)^k} \oint \cdots \oint \Delta(w) \Delta(w^2) \, e^{N \sum_{i=1}^k w_i} \prod_{j=1}^n \left( \sum_{i=1}^k \frac{1}{w_i^{2j}} \right)^{m_j} \prod_{i=1}^k \frac{dw_i}{w_i^{2k}} \\
& =\Delta \left( \frac{d}{dX} \right) \Delta \left( \frac{d}{dY} \right) \prod_{j=1}^n \left( \frac{d}{dt_j} \right)^{m_j} \\
&\qquad \times \left. \frac{1}{(2\pi\i)^k} \oint \cdots \oint \exp \left( \sum_{i=1}^k \left( w_i^2 X_i+w_i Y_i+\sum_{j=1}^n \frac{t_j}{w_i^{2j}} \right) \right) \prod_{i=1}^k \frac{dw_i}{w_i^{2k}} \right|_{\substack{X_i=0, \\ Y_i=N, \\ t_j=0}} \\
& =\prod_{j=1}^n \left( \frac{d}{dt_j} \right)^{m_j} \Delta \left( \frac{d}{dX} \right) \Delta \left( \frac{d}{dY} \right) \left. \prod_{i=1}^k \left( \frac{1}{2\pi\i} \oint_{|w|=1} \exp \left( w^2 X_i+w Y_i+\sum_{j=1}^n \frac{t_j}{w^{2j}} \right) \frac{dw}{w^{2k}} \right) \right|_{\substack{X_i=0, \\ Y_i=N, \\ t_j=0}} \\
& =\prod_{j=1}^n \left( \frac{d}{dt_j} \right)^{m_j} \left. k! \det_{k\times k} \left( \frac{d^{i+j-2}}{dX^{i-1} dY^{j-1}} \frac{1}{2\pi\i} \oint_{|w|=1} \exp \left( w^2 X+w Y+\sum_{l=1}^n \frac{t_l}{w^{2l}} \right) \frac{dw}{w^{2k}} \right) \right|_{\substack{X=0, \\ Y=N, \\ t_l=0}},
\end{align*}
where the last line is by Lemma \ref{Vandermonde 2}. Now,

\begin{align*}
& \frac{d^{i+j-2}}{dX^{i-1} dY^{j-1}} \left. \frac{1}{2\pi\i} \oint \exp \left( w^2 X+w Y+\sum_{l=1}^n \frac{t_l}{w^{2l}} \right) \frac{dw}{w^{2k}} \right|_{\substack{X=0, \\ Y=N}} \\
&\qquad =\frac{1}{2\pi\i} \oint \exp \left( w N+\sum_{l=1}^n \frac{t_l}{w^{2l}} \right) \frac{dw}{w^{2k-2i-j+3}} \\
&\qquad =\frac{N^{2k-2i-j+2}}{2\pi\i} \oint \exp \left( w+\sum_{l=1}^n \frac{N^{2l} t_l}{w^{2l}} \right) \frac{dw}{w^{2k-2i-j+3}}.
\end{align*}
Therefore, we have

\begin{align*}
& \frac{1}{(2\pi\i)^k} \oint \cdots \oint \Delta(w) \Delta(w^2) \, e^{N \sum_{i=1}^k w_i} \prod_{j=1}^n \left( \sum_{i=1}^k \frac{1}{w_i^{2j}} \right)^{m_j} \prod_{i=1}^k \frac{dw_i}{w_i^{2k}} \nonumber \\
= & k! \prod_{j=1}^n \left( \frac{d}{dt_j} \right)^{m_j} \left. \det_{k\times k} \left( \frac{N^{2k-2i-j+2}}{2\pi\i} \oint \exp \left( w+\sum_{l=1}^n \frac{N^{2l} t_l}{w^{2l}} \right) \frac{dw}{w^{2k-2i-j+3}} \right) \right|_{t_j=0} \nonumber \\
= & k! N^{k(k+1)/2} \prod_{j=1}^n \left( \frac{d}{dt_j} \right)^{m_j} \left. \det_{k\times k} \left( \frac{1}{2\pi\i} \oint \exp \left( w+\sum_{l=1}^n \frac{N^{2l} t_l}{w^{2l}} \right) \frac{dw}{w^{2k-2i-j+3}} \right) \right|_{t_j=0} \nonumber \\
= & k! N^{k(k+1)/2+\sum_{j=1}^n 2j m_j} \prod_{j=1}^n \left( \frac{d}{du_j} \right)^{m_j} \left. \det_{k\times k} \left( \frac{1}{2\pi\i} \oint \exp \left( w+\sum_{l=1}^n \frac{u_l}{w^{2l}} \right) \frac{dw}{w^{2k-2i-j+3}} \right) \right|_{u_j=0} \nonumber \\
= &( -1)^{k(k-1)/2} k! N^{k(k+1)/2+\sum_{j=1}^n 2j m_j} \prod_{j=1}^n \left( \frac{d}{du_j} \right)^{m_j} \left. \det_{k\times k} \left( \frac{1}{2\pi\i} \oint \exp \left( w+\sum_{l=1}^n \frac{u_l}{w^{2l}} \right) \frac{dw}{w^{2i-j+1}} \right) \right|_{u_j=0},
\end{align*}
where we have used the fact that $\det_{k\times k} (N^{-2i-j} a_{i,j})=N^{-3k(k+1)/2} \det_{k\times k} (a_{i,j})$. Also, the fourth line follows from the change of variables $u_j=N^{2j} t_j$ and in the last line we have interchanged the rows of the matrix.

Next, by Lemma \ref{hypergeometric integral}, the contour integral appearing in the determinant is

\begin{equation*}
\frac{1}{2\pi\i} \oint \exp \left( w+\sum_{l=1}^n \frac{u_l}{w^{2l}} \right) \frac{dw}{w^{2i-j+1}}=\sum_{l_2, \dots, l_n=0}^{\infty} \left( \prod_{s=2}^n \frac{u_s^{l_s}}{l_s!} \right) g_{2i-j+2\sum_{s=2}^n s l_s} (u_1).
\end{equation*}
We use Lemma \ref{determinant derivative} to carry out the differentiation of the determinant with respect to $u_2, \dots, u_n$. This gives us

\begin{align*}
& \prod_{j=1}^n \left( \frac{d}{du_j} \right)^{m_j} \left. \det_{k\times k} \left( \frac{1}{2\pi\i} \oint \exp \left( w+\sum_{l=1}^n \frac{u_l}{w^{2l}} \right) \frac{dw}{w^{2i-j+1}} \right) \right|_{u_j=0} \\
& =\sum_{\substack{\sum_{i=1}^k m_{s, i}=m_s \\ s=2, \dots, n}} \left( \prod_{s=2}^n \binom{m_s}{m_{s, 1}, \dots, m_{s, k}} \right) \\
&\qquad \times \left( \frac{d}{du_1} \right)^{m_1} \left.  \det_{k\times k} \left( \prod_{s=2}^n \left( \frac{d}{du_s} \right)^{m_{s, i}} \sum_{l_2, \dots, l_n=0}^{\infty} \left( \prod_{s=2}^n \frac{u_s^{l_s}}{l_s!} \right) g_{2i-j+2\sum_{s=2}^n s l_s} (u_1) \right) \right|_{u_j=0} \\
& =\sum_{\substack{\sum_{i=1}^k m_{s, i}=m_s \\ s=2, \dots, n}} \left( \prod_{s=2}^n \binom{m_s}{m_{s, 1}, \dots, m_{s, k}} \right) \left( \frac{d}{du} \right)^{m_1} \left. \det_{k\times k} \left( g_{2i-j+2 \sum_{s=2}^n s m_{s, i}}(u) \right) \right|_{u=0}.
\end{align*}
Putting it all together yields (\ref{integral 1}). The proof of (\ref{integral 2}) is similar.
\end{proof}

\begin{proposition} \label{2nd integral prop}
Let $k \geq 1$ and $m_j$ be integers for $j=1, \dots, k$. Then we have

\begin{equation*}
\frac{1}{(2\pi\i)^k} \oint \cdots \oint_{|w_i|=1} \frac{\Delta(w) \Delta(w^2) \, e^{N \sum_{i=1}^k w_i}}{\prod_{j=1}^k w_j^{2k+m_j}} \prod_{i=1}^k dw_i=\sum_{\mu \in S_k} \det_{k \times k} \left( \frac{N^{2k+m_{\mu(i)}-2i-j+2}}{\Gamma(2k+m_{\mu(i)}-2i-j+3)} \right).
\end{equation*}
\end{proposition}

\begin{proof}
As in the proof of Proposition \ref{1st integral prop}, we write

\begin{equation*}
\Delta(w) \Delta(w^2) \, e^{N \sum_{i=1}^k w_i}=\Delta \left( \frac{d}{dX} \right) \Delta \left( \frac{d}{dY} \right) \left. \exp \left( \sum_{i=1}^k w_i^2 X_i+w_i Y_i \right) \right|_{\substack{X_i=0, \\ Y_i=N}}.
\end{equation*}
Then, we have that

\begin{equation*}
\frac{1}{(2\pi\i)^k} \oint \cdots \oint \frac{\Delta(w) \Delta(w^2) \, e^{N \sum_{i=1}^k w_i}}{\prod_{j=1}^k w_j^{2k+m_j}} \prod_{i=1}^k dw_i=\Delta \left( \frac{d}{dX} \right) \Delta \left( \frac{d}{dY} \right) \left. \prod_{i=1}^k f_i(X_i, Y_i) \right|_{\substack{X_i=0, \\ Y_i=N}},
\end{equation*}
where

\begin{equation*}
f_i(X_i, Y_i)=\frac{1}{2\pi\i} \oint_{|w|=1} \frac{e^{(w^2 X_i+w Y_i)}}{w^{2k+m_i}} dw.
\end{equation*}
So, by Lemma \ref{Vandermonde 2}, we have

\begin{equation*}
\frac{1}{(2\pi\i)^k} \oint \cdots \oint \frac{\Delta(w) \Delta(w^2) \, e^{N \sum_{i=1}^k w_i}}{\prod_{j=1}^k w_j^{2k+m_j}} \prod_{i=1}^k dw_i=\sum_{\mu \in S_k} \left. \det_{k \times k} \left( \frac{d^{i+j-2}}{dX^{i-1} dY^{j-1}} f_{\mu(i)}(X,Y) \right) \right|_{\substack{X=0, \\ Y=N}}.
\end{equation*}
Now,

\begin{align*}
\frac{d^{i+j-2}}{dX^{i-1} dY^{j-1}} \left. f_{\mu(i)}(X,Y) \right|_{\substack{X=0, \\ Y=N}}=\frac{1}{2\pi\i} \oint_{|w|=1} \frac{e^{N w}}{w^{2k+m_{\mu(i)}-2i-j+3}} dw=\frac{N^{2k+m_{\mu(i)}-2i-j+2}}{\Gamma(2k+m_{\mu(i)}-2i-j+3)},
\end{align*}
and the proposition follows.
\end{proof}

Lastly, we include a lemma that allows us to explicitly evaluate certain determinants whose entries are reciprocals of the Gamma function.

\begin{lemma} \label{Gamma determinant}
Let $k \geq 1$ and $m_j \geq 0$ be integers for $j=1, \dots, k$. Then, we have

\begin{equation*}
\det_{k \times k} \left( \frac{1}{\Gamma(2k+m_i-2i-j+2)} \right)=\prod_{j=1}^k \frac{1}{(2k+m_j-2j)!} \prod_{1 \leq i,j \leq k} (m_j-m_i-2j+2i).
\end{equation*}
\end{lemma}

\begin{proof}
With our notation, equation (4.13) in \cite{Normand04} can be written as

\begin{equation*}
\det_{k \times k} \left( \frac{1}{\Gamma(z_i-j+1)} \right)=\frac{\Delta(z_1, \dots, z_k)}{\prod_{j=1}^k \Gamma(z_j)}.
\end{equation*}
We take $z_i=2k+m_i-2i+1$ for $i=1, \dots, k$. Then, we have

\begin{align*}
\det_{k \times k} \left( \frac{1}{\Gamma(2k+m_i-2i-j+2)} \right) &=\prod_{i=1}^k \Gamma(2k+m_i-2i+1)^{-1} \prod_{1 \leq i<j \leq k} (m_j-2j-m_i+2i).
\end{align*}
Since $2k+m_i-2i+1 \geq 1$ for $1 \leq i \leq k$, we have that $\Gamma(2k+m_i-2i+1)=(2k+m_i-2i)!$ which completes the proof.

\end{proof}

\section{Proofs of the main results}

In this section we will present the proofs of our main results. Our strategy is to obtain the joint moments by differentiating the corresponding shifted moments with respect to the shifts. Indeed, one may check by induction that for $G(2N)\in\{Sp(2N), \, SO(2N), \, O^-(2N)\}$, we have that

\begin{align*}
& \int_{G(2N)} \left( \Lambda_A^{(n_1)}(1) \right)^{k_1} \left( \Lambda_A^{(n_2)}(1) \right)^{k_2} dA \\
&\qquad =\prod_{j=1}^{k_1} \left( \frac{d}{d\alpha_j} \right)^{n_1} \prod_{j=k_1+1}^k \left( \frac{d}{d\alpha_j} \right)^{n_2} \left. I(G(2N);e^{-\alpha_1},\dots,e^{-\alpha_k}) \right|_{\alpha_j=0} \left( 1+O(N^{-1}) \right),
\end{align*}
where $k=k_1+k_2$. Also, the error terms in Lemmas \ref{Sp shifted moments}-\ref{O^- shifted moments} are uniform in $\alpha$ so we obtain an asymptotic formula after performing the differentiation.

\subsection{The unitary symplectic group $Sp(2N)$}

\begin{proof}[Proof of Theorem \ref{Sp theorem 1}]

By the above argument and Lemma \ref{Sp shifted moments}, we have that

\begin{equation} \label{start of proofs}
\int_{Sp(2N)} \left( \Lambda_A^{(n_1)}(1) \right)^{k_1} \left( \Lambda_A^{(n_2)}(1) \right)^{k_2} dA=\frac{(-1)^{k(k-1)/2}}{k!} J_{k_1, k_2}^{Sp}(n_1, n_2) \left( 1+O(N^{-1}) \right), 
\end{equation}
where

\begin{align} \label{J integral def}
& J_{k_1, k_2}^{Sp}(n_1, n_2) \nonumber \\
& =\prod_{j=1}^{k_1} \left( \frac{d}{d\alpha_j} \right)^{n_1} \prod_{j=k_1+1}^k \left( \frac{d}{d\alpha_j} \right)^{n_2} \left. \frac{1}{(2\pi\i)^k} \oint \cdots \oint_{|w_i|=1} \frac{\Delta(w) \Delta(w^2) \, e^{N \sum_{i=1}^k (w_i-\alpha_i)}}{\prod_{1\leq i,j\leq k} (w_i^2-\alpha_j^2)} \, \prod_{i=1}^k dw_i \right|_{\alpha_j=0}.
\end{align}
We use Lemma \ref{2nd derivative} to carry out the differentiation and obtain

\begin{align*}
& J_{k_1, k_2}^{Sp}(n_1, n_2) \\ 
& =\frac{1}{(2\pi\i)^k} \oint \cdots \oint  \Bigg( \sum_{\substack{a_1+2a_2+\cdots+n_1 a_{n_1}=n_1 \\ a_3=a_5=\cdots=0}} \frac{n_1!}{a_1!\cdots a_{n_1}!} (-N)^{a_1} \prod_{j=1}^{[n_1/2]} \left( \frac{1}{j} \sum_{i=1}^k \frac{1}{w_i^{2j}} \right)^{a_{2j}} \Bigg)^{k_1} \\
&\quad \times \Bigg( \sum_{\substack{b_1+2b_2+\cdots+n_2 b_{n_2}=n_2 \\ b_3=b_5=\cdots=0}} \frac{n_2!}{b_1!\cdots b_{n_2}!} (-N)^{b_1} \prod_{j=1}^{[n_2/2]} \left( \frac{1}{j} \sum_{i=1}^k \frac{1}{w_i^{2j}} \right)^{b_{2j}} \Bigg)^{k_2} \Delta(w) \Delta(w^2) \, e^{N \sum_{i=1}^k w_i} \prod_{i=1}^k \frac{dw_i}{w_i^{2k}}.
\end{align*}
Recall the definition of the tuples ${\bf a}_i$ and ${\bf b}_i$ defined in the statement of the theorem. Then, we can expand the brackets in the integrand of $J_{k_1, k_2}^{Sp}(n_1, n_2)$ as

\begin{align*}
& \Bigg( \sum_{\substack{a_1+2a_2+\cdots+n_1 a_{n_1}=n_1 \\ a_3=a_5=\cdots=0}} \frac{n_1!}{a_1!\cdots a_{n_1}!} (-N)^{a_1} \prod_{j=1}^{[n_1/2]} \left( \frac{1}{j} \sum_{i=1}^k \frac{1}{w_i^{2j}} \right)^{a_{2j}} \Bigg)^{k_1} \\
& =\sum_{u_1+\cdots+u_P=k_1} \binom{k_1}{u_1, \dots, u_P} \frac{(n_1)!^{k_1}}{\prod_{i=1}^P ({\bf a}_i!)^{u_i}} (-N)^{\sum_{i=1}^P u_i a_{i,0}} \prod_{j=1}^{[n_1/2]} \left( \frac{1}{j} \sum_{l=1}^k \frac{1}{w_l^{2j}} \right)^{\sum_{i=1}^P u_i a_{i,j}},
\end{align*}
with a similar expression for the bracket to the power of $k_2$ in the integrand. Using these expansions, our expression for $J_{k_1, k_2}^{Sp}(n_1, n_2)$ becomes

\begin{align*}
J_{k_1, k_2}^{Sp}(n_1, n_2) & =\sum_{u_1+\cdots+u_P=k_1} \binom{k_1}{u_1, \dots, u_P} \frac{(n_1)!^{k_1}}{\prod_{i=1}^P ({\bf a}_i!)^{u_i} (\prod_{j=1}^{[n_1/2]} j^{\sum_{i=1}^P u_i a_{i, j}})} (-N)^{\sum_{i=1}^P u_i a_{i,0}} \\
& \times \sum_{v_1+\cdots+v_Q=k_2} \binom{k_2}{v_1, \dots, v_Q} \frac{(n_2)!^{k_2}}{\prod_{i=1}^Q ({\bf b}_i!)^{v_i} (\prod_{j=1}^{[n_2/2]} j^{\sum_{i=1}^Q v_i b_{i, j}})} (-N)^{\sum_{i=1}^Q v_i b_{i,0}} \\
& \times \frac{1}{(2\pi\i)^k} \oint \cdots \oint_{|w_i|=1} \Delta(w) \Delta(w^2) \, e^{N \sum_{i=1}^k w_i} \prod_{j=1}^{[n_2/2]} \left( \sum_{l=1}^k \frac{1}{w_l^{2j}} \right)^{m_j} \prod_{i=1}^k \frac{dw_i}{w_i^{2k}},
\end{align*}
where $m_j=\sum_{i=1}^P u_i a_{i, j}+\sum_{i=1}^Q v_i b_{i,j}$ for $j=1, \dots, [n_1/2]$ and $m_j=\sum_{i=1}^Q v_i b_{i,j}$ for $j=[n_1/2]+1, \dots, [n_2/2]$. We now apply Proposition \ref{1st integral prop} to the contour integral above with these $m_j$ which gives us that

\begin{align*}
J_{k_1, k_2}^{Sp}(n_1, n_2) & =(-1)^{k(k-1)/2} k! N^{k(k+1)/2}\sum_{u_1+\cdots+u_P=k_1} \binom{k_1}{u_1, \dots, u_P} \frac{(n_1)!^{k_1} (-N)^{\sum_{i=1}^P u_i a_{i,0}}}{\prod_{i=1}^P ({\bf a}_i!)^{u_i} (\prod_{j=1}^{[n_1/2]} j^{\sum_{i=1}^P u_i a_{i, j}})} \\
& \times \sum_{v_1+\cdots+v_Q=k_2} \binom{k_2}{v_1, \dots, v_Q} \frac{(n_2)!^{k_2} (-N)^{\sum_{i=1}^Q v_i b_{i,0}}}{\prod_{i=1}^Q ({\bf b}_i!)^{v_i} (\prod_{j=1}^{[n_2/2]} j^{\sum_{i=1}^Q v_i b_{i, j}})} \cdot N^{2 \sum_{j=1}^{[n_2/2]} j m_j} \\
& \times \sum_{\substack{\sum_{i=1}^k m_{s, i}=m_s \\ s=2, \dots, n}} \left( \prod_{s=2}^n \binom{m_s}{m_{s, 1}, \dots, m_{s, k}} \right) \left( \frac{d}{du} \right)^{m_1} \left. \det_{k\times k} \left( g_{2i-j+2 \sum_{s=2}^n s m_{s, i}}(u) \right) \right|_{u=0}.
\end{align*}
Using the definition of $m_j$, we compute the power of $N$ in the summand as

\begin{align*}
\sum_{i=1}^P u_i a_{i, 0}+\sum_{i=1}^Q v_i b_{i, 0}+2 \sum_{j=1}^{[n_2/2]} jm_j & =\sum_{i=1}^P u_i \left( a_{i, 0}+2 \sum_{j=1}^{[n_1/2]} j a_{i, j} \right)+\sum_{i=1}^Q v_i \left( b_{i, 0}+2 \sum_{j=1}^{[n_2/2]} j b_{i, j} \right) \\
& =n_1 \sum_{i=1}^P u_i+n_2 \sum_{i=1}^Q v_i \\
& =k_1 n_1+k_2 n_2.
\end{align*}
Also, since $a_{i, 0} \equiv n_1 \ (\textrm{mod} \ 2)$ and $b_{i, 0} \equiv n_2 \ (\textrm{mod} \ 2)$ for all $i$, the factor of $(-1)$ in the summand is

\begin{equation*}
(-1)^{\sum_{i=1}^P u_i a_{i, 0}+\sum_{i=1}^Q v_i b_{i, 0}}=(-1)^{n_1 \sum_{i=1}^P u_i+n_2 \sum_{i=1}^Q v_i}=(-1)^{k_1 n_1+k_2 n_2}. 
\end{equation*}
Combining these two observations with our final expression for $J_{k_1, k_2}^{Sp}(n_1, n_2)$ completes the proof.

\end{proof}

\begin{proof}[Proof of Theorem \ref{Sp theorem 2}]

We begin as in the proof of Theorem \ref{Sp theorem 1} with (\ref{start of proofs}) and (\ref{J integral def}). We use Lemma \ref{1st derivative} for the derivatives in this case which gives us that

\begin{align*}
J_{k_1, k_2}^{Sp}(n_1, n_2) & =\frac{1}{(2\pi\i)^k} \oint \cdots \oint_{|w_i|=1}  \Bigg( \sum_{m=0}^{n_1} \binom{n_1}{m} (-N)^{n_1-m} \, m! \sum_{\substack{l_1+\cdots+l_k=m \\ l_j \ \textrm{even}}} \prod_{i=1}^k \frac{1}{w_i^{l_i}} \Bigg)^{k_1} \\
& \times \Bigg( \sum_{m=0}^{n_2} \binom{n_2}{m} (-N)^{n_2-m} \, m! \sum_{\substack{l_1+\cdots+l_k=m \\ l_j \ \textrm{even}}} \prod_{i=1}^k \frac{1}{w_i^{l_i}} \Bigg)^{k_2} \Delta(w) \Delta(w^2) \, e^{N \sum_{i=1}^k w_i} \prod_{i=1}^k \frac{dw_i}{w_i^{2k}}.
\end{align*}
Rather than expand the brackets in the integrand, we write them as

\begin{align*}
& \Bigg( \sum_{m=0}^{n_1} \binom{n_1}{m} (-N)^{n_1-m} \, m! \sum_{\substack{l_1+\cdots+l_k=m \\ l_j \ \textrm{even}}} \prod_{i=1}^k \frac{1}{w_i^{l_i}} \Bigg)^{k_1} \\
& = \Bigg( \sum_{\substack{\sum_{j=1}^k l_j \leq n_1 \\ l_j \ \textrm{even}}} (-N)^{n_1-\sum_{j=1}^k l_j} \binom{n_1}{\sum_{j=1}^k l_j} \left( \sum_{j=1}^k l_j \right)! \prod_{j=1}^k \frac{1}{w_j^{l_j}} \Bigg)^{k_1} \\
& =\sum_{\substack{2 \sum_{j=1}^k l_{i, j} \leq n_1 \\ i=1, \dots, k_1}} \prod_{i=1}^{k_1} \left( (-N)^{n_1-2 \sum_{j=1}^k l_{i, j}} \binom{n_1}{2 \sum_{j=1}^k l_{i, j}} \left( 2 \sum_{j=1}^k l_{i, j} \right)! \right) \prod_{j=1}^k \frac{1}{w_j^{2 \sum_{i=1}^{k_1} l_{i, j}}},
\end{align*}
with a similar expression for the second bracket to the $k_2$. We then have that

\begin{align*}
J_{k_1, k_2}^{Sp}(n_1, n_2) & =\sum_{\substack{2 \sum_{j=1}^k l_{i, j} \leq n_1 \\ i=1, \dots, k_1}} \prod_{i=1}^{k_1} \left( (-N)^{n_1-2 \sum_{j=1}^k l_{i, j}} \binom{n_1}{2 \sum_{j=1}^k l_{i, j}} \left( 2 \sum_{j=1}^k l_{i, j} \right)! \right) \\
&\qquad \times \sum_{\substack{2 \sum_{j=1}^k m_{i, j} \leq n_2 \\ i=1, \dots, k_2}} \prod_{i=1}^{k_2} \left( (-N)^{n_2-2 \sum_{j=1}^k m_{i, j}} \binom{n_2}{2 \sum_{j=1}^k m_{i, j}} \left( 2 \sum_{j=1}^k m_{i, j} \right)! \right) \\
&\qquad \times \frac{1}{(2\pi\i)^k} \oint \cdots \oint_{|w_i|=1} \frac{\Delta(w) \Delta(w^2) \, e^{N \sum_{i=1}^k w_i}}{\prod_{j=1}^k w_j^{2k+2\sum_{i=1}^{k_1} l_{i, j}+2\sum_{i=1}^{k_2} m_{i, j}}} \prod_{i=1}^k dw_i \\
& =\sum_{\substack{2 \sum_{j=1}^k l_{i, j} \leq n_1 \\ i=1, \dots, k_1}} \sum_{\substack{2 \sum_{j=1}^k m_{i, j} \leq n_2 \\ i=1, \dots, k_2}} (n_1!)^{k_1} (n_2!)^{k_2} (-N)^{k_1 n_1+k_2 n_2-2 \sum_{j=1}^k (\sum_{i=1}^{k_1} l_{i, j}+\sum_{i=1}^{k_2} m_{i j})} \\
&\qquad \times \left( \prod_{i=1}^{k_1} \frac{1}{(n_1-2 \sum_{j=1}^k l_{i ,j})!} \right) \left( \prod_{i=1}^{k_2} \frac{1}{(n_2-2 \sum_{j=1}^k m_{i, j})!} \right) \\
&\qquad \times \frac{1}{(2\pi\i)^k} \oint \cdots \oint_{|w_i|=1} \frac{\Delta(w) \Delta(w^2) \, e^{N \sum_{i=1}^k w_i}}{\prod_{j=1}^k w_j^{2k+2\sum_{i=1}^{k_1} l_{i, j}+2\sum_{i=1}^{k_2} m_{i, j}}} \prod_{i=1}^k dw_i.
\end{align*}
We set $V_j=2 \sum_{i=1}^{k_1} l_{i, j}+2 \sum_{i=1}^{k_2} m_{i, j}$ for $j=1, \dots, k$. Then, by Proposition \ref{2nd integral prop}, the contour integral in the last line above is equal to

\begin{align*}
\sum_{\mu \in S_k} \det_{k \times k} \left( \frac{N^{2k+V_{\mu(i)}-2i-j+2}}{\Gamma(2k+V_{\mu(i)}-2i-j+3)} \right) & =N^{k(k+1)/2+\sum_{j=1}^k V_j} \sum_{\mu \in S_k} \det_{k \times k} \left( \frac{1}{\Gamma(2k+V_{\mu(i)}-2i-j+3)} \right).
\end{align*}
Hence, our expression for $J_{k_1, k_2}^{Sp}(n_1, n_2)$ becomes

\begin{align*}
J_{k_1, k_2}^{Sp}(n_1, n_2) & =(-1)^{k_1 n_1+k_2 n_2} (n_1!)^{k_1} (n_2!)^{k_2} N^{k(k+1)/2+k_1 n_1+k_2 n_2} \\
& \times \sum_{\mu \in S_k} \sum_{\substack{2 \sum_{j=1}^k l_{i, j} \leq n_1 \\ i=1, \dots, k_1}} \sum_{\substack{2 \sum_{j=1}^k m_{i, j} \leq n_2 \\ i=1, \dots, k_2}} \left( \prod_{i=1}^{k_1} \frac{1}{(n_1-2 \sum_{j=1}^k l_{i ,j})!} \right) \left( \prod_{i=1}^{k_2} \frac{1}{(n_2-2 \sum_{j=1}^k m_{i, j})!} \right) \\
& \times \det_{k \times k} \left( \frac{1}{\Gamma(2k+V_{\mu(i)}-2i-j+3)} \right).
\end{align*}
Now, by an argument similar to that given at the end of the proof of \cite[Theorem 25]{KW23a}, we have that the sums over $l_{i, j}$ and $m_{i, j}$ do not depend on the choice of permutation $\mu$. Thus, we may take $\mu$ to be the identity and replace the sum over $\mu \in S_k$ by $k!$. To finish, we apply Lemma \ref{Gamma determinant} with $m_j=V_j+1$ to the last determinant which gives us

\begin{align*}
J_{k_1, k_2}^{Sp}(n_1, n_2) & =(-1)^{k_1 n_1+k_2 n_2} (n_1!)^{k_1} (n_2!)^{k_2} N^{k(k+1)/2+k_1 n_1+k_2 n_2} \\
& \times \sum_{\mu \in S_k} \sum_{\substack{2 \sum_{j=1}^k l_{i, j} \leq n_1 \\ i=1, \dots, k_1}} \sum_{\substack{2 \sum_{j=1}^k m_{i, j} \leq n_2 \\ i=1, \dots, k_2}} \left( \prod_{i=1}^{k_1} \frac{1}{(n_1-2 \sum_{j=1}^k l_{i ,j})!} \right) \left( \prod_{i=1}^{k_2} \frac{1}{(n_2-2 \sum_{j=1}^k m_{i, j})!} \right) \\
& \times \prod_{j=1}^k \frac{1}{(2k+V_j-2j+1)!} \prod_{1 \leq i<j \leq k} (V_j-V_i-2j+2i).
\end{align*}
The theorem follows.

\end{proof}

\subsection{The special orthogonal group $SO(2N)$ and $O^-(2N)$}

\begin{proof}[Proof of Theorem \ref{SO theorem 1}]

The proof is similar to that of Theorem \ref{Sp theorem 1} but we now use Lemma \ref{SO shifted moments} for the shifted moments. We again use Lemma \ref{2nd derivative} for the derivatives and apply (\ref{integral 2}) in Proposition \ref{1st integral prop} to the resulting contour integral.

\end{proof}

\begin{proof}[Proof of Theorem \ref{SO theorem 2}]

We follow the proof of Theorem \ref{Sp theorem 2} using Lemma \ref{SO shifted moments} for the shifted moments and Lemma \ref{1st derivative} for the derivatives. We then use Proposition \ref{2nd integral prop} with $m_j=V_j-1$ for the contour integral and conclude the proof similarly using Lemma \ref{Gamma determinant}.

\end{proof}

\begin{proof}[Proof of Theorem \ref{O^- theorem}]

Using the argument at the beginning of the section and Lemma \ref{O^- shifted moments}, we have that

\begin{equation*}
\int_{O^-(2N)} \left( \Lambda_A^{(n_1)}(1) \right)^{k_1} \left( \Lambda_A^{(n_2)}(1) \right)^{k_2} dA=\frac{(-1)^{k(k-1)/2} 2^k}{k!} J_{k_1, k_2}^{O^-(2N)}(n_1, n_2) \left( 1+O(N^{-1}) \right),
\end{equation*}
where

\begin{align*}
& J_{k_1, k_2}^{O^-}(n_1, n_2) \\
& =\prod_{j=1}^{k_1} \left( \frac{d}{d\alpha_j} \right)^{n_1} \prod_{j=k_1+1}^k \left( \frac{d}{d\alpha_j} \right)^{n_2} \left. \frac{1}{(2\pi\i)^k} \oint \cdots \oint \frac{\Delta(w) \Delta(w^2) (\prod_{i=1}^k \alpha_i) \, e^{N \sum_{i=1}^k (w_i+\alpha_i)}}{\prod_{1\leq i,j\leq k} (w_i^2-\alpha_j^2)} \, \prod_{i=1}^k dw_i \right|_{\alpha_j=0}.
\end{align*}
For the derivatives, we use the fact that for $n \geq 1$,

\begin{align*}
\frac{d^n}{d\alpha^n} \left. \frac{\alpha e^{N \alpha}}{\prod_{i=1}^k (w_i^2-\alpha^2)} \right|_{\alpha=0}=n \frac{d^{n-1}}{d\alpha^{n-1}} \left. \frac{e^{N \alpha}}{\prod_{i=1}^k (w_i^2-\alpha^2)} \right|_{\alpha=0}.
\end{align*}
Hence, we have that

\begin{align*}
& J_{k_1, k_2}^{O^-}(n_1, n_2)=n_1^{k_1} n_2^{k_2} \\
&\qquad \times \prod_{j=1}^{k_1} \left( \frac{d}{d\alpha_j} \right)^{n_1-1} \prod_{j=k_1+1}^k \left( \frac{d}{d\alpha_j} \right)^{n_2-1} \left. \frac{1}{(2\pi\i)^k} \oint \cdots \oint \frac{\Delta(w) \Delta(w^2) \, e^{N \sum_{i=1}^k (w_i+\alpha_i)}}{\prod_{1\leq i,j\leq k} (w_i^2-\alpha_j^2)} \, \prod_{i=1}^k dw_i \right|_{\alpha_j=0}.
\end{align*}
This integral expression is very similar to the expression for $J_{k_1, k_2}^{Sp(2N)}(n_1-1, n_2-1)$ given in (\ref{J integral def}) and so we can proceed as in the proof of Theorem \ref{Sp theorem 1} or \ref{Sp theorem 2}, simply replacing $(-N)$ by $N$ when we use Lemma \ref{1st derivative} or \ref{2nd derivative}. In either case, we obtain the statement of the theorem.

\end{proof}

\subsection{Proof of Proposition \ref{b(0,n) prop}}

We conclude this setion by proving Proposition \ref{b(0,n) prop}. Let $n\geq 1$ be an integer. Then, by Theorem \ref{Sp theorem 2}, we have that

\begin{align*}
b_{0,1}^{Sp}(0,n) & =\frac{(-1)^n n!}{2^{n+1}} \sum_{2l \leq n} \frac{1}{(n-2l)! (2l+1)!} \\
& =\frac{(-1)^n n!}{2^{n+1} (n+1)!} \sum_{2l \leq n} \binom{n+1}{2l+1} \\
& =\frac{(-1)^n}{2^{n+1} (n+1)} \left( \sum_{2l \leq n} \binom{n}{2l}+\sum_{2l \leq n-1} \binom{n}{2l+1} \right) \\
& =\frac{(-1)^n}{2^{n+1} (n+1)} \sum_{l=0}^n \binom{n}{l} \\
& =\frac{(-1)^n}{2 (n+1)},
\end{align*}
where we have used standard properties of the binomial coefficient. In the same manner, by Theorem \ref{SO theorem 2}, we have

\begin{align*}
b_{0,1}^{SO}(0,n) & =2^{1-n} n! \sum_{2l \leq n} \frac{1}{(n-2l)! (2l)!} \\
& =2^{1-n} \sum_{2l \leq n} \binom{n}{2l} \\
& =2^{1-n} \left( \sum_{2l \leq n-1} \binom{n-1}{2l}+\sum_{2l \leq n} \binom{n-1}{2l-1} \right) \\
& =2^{1-n} \sum_{l=0}^{n-1} \binom{n-1}{l} \\
& =1.
\end{align*}

\section{Numerical results}

Below we give some numerical values for $b_{k_1, k_2}^{Sp}(n_1, n_2)$ and $b_{k_1, k_2}^{SO}(n_1, n_2)$. Values of $b_{k_1, k_2}^{O^-}(n_1, n_2)$ follow from Theorem \ref{O^- theorem} so are omitted. Numerical values for $b_{0,k}^{Sp}(0,2)$ and  $b_{0,k}^{SO}(0,2)$ for $k \leq 10$ are given in \cite[Section 4]{ABPRW14}. 

The following are $b_{0,k}^{Sp}(0,3)$ for $k=1,\dots, 4$:

\begin{equation*}
-\frac{1}{2^3}
\end{equation*}

\begin{equation*}
\frac{23}{2^7 \cdot 3 \cdot 5 \cdot 7}
\end{equation*}

\begin{equation*}
-\frac{1}{2^8 \cdot 5^2 \cdot 7 \cdot 11}
\end{equation*}

\begin{equation*}
\frac{233}{2^{18} \cdot 3^4 \cdot 5^3 \cdot 7^2 \cdot 11}.
\end{equation*}
$b_{0, k}^{Sp}(0,4)$ for $k=1, \dots, 4$:



\begin{equation*}
\frac{1}{2 \cdot 5}
\end{equation*}

\begin{equation*}
\frac{251}{2^4 \cdot 3^2 \cdot 5^2 \cdot 7 \cdot 11}
\end{equation*}

\begin{equation*}
\frac{89 \cdot 13103}{2^9 \cdot 3^5 \cdot 5^3 \cdot 7^2 \cdot 11 \cdot 13 \cdot 17}
\end{equation*}

\begin{equation*}
\frac{1627 \cdot 693731}{2^{10} \cdot 3^5 \cdot 5^5 \cdot 7^3 \cdot 11^2 \cdot 13^2 \cdot 17 \cdot 19 \cdot 23}.
\end{equation*}
We also have $b_{1,1}^{Sp}(n_1,1)$ for $n_1=0,1$:

\begin{equation*}
-\frac{1}{48}, \ \frac{1}{96}.
\end{equation*}
$b_{1,1}^{Sp}(n_1,2)$ for $n_1=0,1,2$:

\begin{equation*}
\frac{1}{80}, \ -\frac{1}{160}, \ \frac{19}{5040}.
\end{equation*}
$b_{1,1}^{Sp}(n_1,3)$ for $n_1=0,1,2,3$:

\begin{equation*}
-\frac{1}{120}, \ \frac{1}{240}, \ -\frac{17}{6720}, \ \frac{23}{13440}.
\end{equation*}
$b_{1,2}^{Sp}(n_1,1)$ for $n_1=0,1$:

\begin{equation*}
\frac{1}{11520}, \ -\frac{1}{23040}.
\end{equation*}
$b_{1,2}^{Sp}(n_1,2)$ for $n_1=0,1,2$:

\begin{equation*}
\frac{103}{3628800}, \ -\frac{103}{7257600}, \ \frac{487}{59875200}.
\end{equation*}
$b_{1,2}^{Sp}(n_1,3)$ for $n_1=0,1,2,3$:

\begin{equation*}
\frac{1}{89600}, \ -\frac{1}{179200}, \ \frac{19}{5913600}, \ -\frac{1}{492800}.
\end{equation*}

The following are $b_{0,k}^{SO}(0,3)$ for $k=1,2,3,4$:

\begin{equation*}
1
\end{equation*}

\begin{equation*}
\frac{3}{2^2 \cdot 5}
\end{equation*}

\begin{equation*}
\frac{1}{2^4 \cdot 3 \cdot 7}
\end{equation*}

\begin{equation*}
\frac{1613}{2^9 \cdot 3 \cdot 5^2 \cdot 7^2 \cdot 11 \cdot 13}.
\end{equation*}
$b_{0,k}^{SO}(0,4)$ for $k=1,2,3,4$:

\begin{equation*}
1
\end{equation*}

\begin{equation*}
\frac{71}{2 \cdot 3^2 \cdot 5 \cdot 7}
\end{equation*}

\begin{equation*}
\frac{23 \cdot 2657}{2 \cdot 3^3 \cdot 5^3 \cdot 7^2 \cdot 11 \cdot 13}
\end{equation*}

\begin{equation*}
\frac{7159 \cdot 316201}{2^6 \cdot 3^5 \cdot 5^4 \cdot 7^3 \cdot 11^2 \cdot 13 \cdot 17 \cdot 19}.
\end{equation*}
We also have $b_{1,1}^{SO}(n_1,1)$ for $n_1=0,1$:

\begin{equation*}
1, \ \frac{1}{2}.
\end{equation*}
$b_{1,1}^{SO}(n_1,2)$ for $n_1=0,1,2$:

\begin{equation*}
\frac{2}{3}, \ \frac{1}{3}, \ \frac{7}{30}.
\end{equation*}
$b_{1,1}^{SO}(n_1,3)$ for $n_1=0,1,2,3$:

\begin{equation*}
\frac{1}{2}, \ \frac{1}{4}, \ \frac{11}{60}, \ \frac{3}{20}.
\end{equation*}
$b_{1,2}^{SO}(n_1,1)$ for $n_1=0,1$:

\begin{equation*}
\frac{1}{12}, \ \frac{1}{24}.
\end{equation*}
$b_{1,2}^{SO}(n_1,2)$ for $n_1=0,1,2$:

\begin{equation*}
\frac{19}{630}, \ \frac{19}{1260}, \ \frac{26}{2835}.
\end{equation*}
$b_{1,2}^{SO}(n_1,3)$ for $n_1=0,1,2,3$:

\begin{equation*}
\frac{23}{1680}, \ \frac{23}{3360}, \ \frac{43}{10080}, \ \frac{1}{336}.
\end{equation*}

\vspace{0.5cm}

\noindent \textit{Acknowledgments.}
The first author is grateful to the Leverhulme Trust (RPG-2017-320) for the support through the research project grant ``Moments of $L$-functions in Function Fields and Random Matrix Theory". The research of the second author is supported by an EPSRC Standard Research Studentship (DTP) at the University of Exeter.

\end{document}